\documentclass{article}
\usepackage{graphicx} 

\usepackage{amsmath,amssymb,amsfonts,amsthm, mathtools}

\usepackage{xcolor}
\usepackage{hyperref}

\usepackage[capitalize]{cleveref}

\usepackage{geometry}

\usepackage{enumitem}

\usepackage{physics,bm}

\usepackage{thm-restate}

\newtheorem{theorem}{Theorem}[section]

\newtheorem{lemma}[theorem]{Lemma}
\newtheorem{remark}{Remark}
{\bf}{\it}

\renewcommand{\i}{\ensuremath{\text{\normalfont I}}}
\newcommand{\ii}{\ensuremath{\text{\normalfont I\!I}}}
\newcommand{\iii}{\ensuremath{\text{\normalfont I\!I\!I}}}

\newcommand{\cL}{\mathcal{L}}
\newcommand{\cM}{\mathcal{M}}

\newcommand{\cT}{\mathcal{T}}

\title{Quantum Regression Theory and Efficient Computation of Response Functions for Non-Markovian Open Systems}

\author{Xiantao Li \thanks{Xiantao.Li@psu.edu}\\
Department of Mathematics, \\ The Pennsylvania State University, USA \\
 Chunhao Wang \thanks{cwang@psu.edu} \\
Department of Computer Science and Engineering, \\
Pennsylvania State University, USA}

\date{\empty}

\begin{document}
\maketitle

\begin{abstract}
 Linear response functions are a cornerstone concept in physics as they enable efficient estimation of many dynamical properties. In addition to predicting dynamics of observables under perturbations without resimulating the system, these response functions lead to electric conductivity, magnetic susceptibility, dielectric constants, etc.  Estimating two-time correlation functions is a key ingredient of measuring linear response functions. However,  for open quantum systems, simulating the reduced density operator with a quantum master equation only yields \emph{one-point} observables and is insufficient for this task. In this paper, we develop a memoryless, system-only formulation of two-point correlations for open quantum systems that extends the standard quantum regression theorem (QRT) beyond the Markov limit. We further incorporate the spectral property of the bath and express the time propagators in the response function as the memoryless generators in Lindblad-type forms. The resulting expressions recast the total response function into evolutions generated by time-dependent Hamiltonian and Lindblad primitives together with the more challenging propagation of commutators and anti-commutators. In addition to the derivation of the new QRT,  we present quantum algorithms for these primitives and obtain an estimator for two-time correlations whose cost scales poly-logarithmically in the system dimension and $1/\epsilon^{1.25}$ in the target accuracy $\epsilon$. The framework removes the separability (Born-Markov) assumption and offers a pathway to efficient computation of nonequilibrium properties from open quantum systems.

\end{abstract}

\section{Introduction}\label{sec:intro}

Estimating physical quantities is the ultimate goal of simulating quantum dynamics. In the regime of closed quantum systems, there are plenty of efficient quantum algorithmic tools for this task, which also generalize to Markovian open quantum systems. However, for some quantities related to the out-of-equilibrium behavior of general open quantum systems, especially in the presence of external perturbations, it is not known how to estimate them within known formulations and methods.
Linear-response theory (LRT) provides a route by which the first-order change of an observable can be written in terms of \emph{unperturbed} two-time correlation functions of the unperturbed dynamics \cite{fetter2012quantum}, e.g., the electric current induced by a biased voltage.   
For an \emph{open} quantum system, whose total Hamiltonian is 
\[
H_{\mathrm{tot}} = H_S\otimes I_B + I_S\otimes H_B + H_{SB},
\]
these correlations must be evaluated for operators that act only on the system subspace while the composite system evolves unitarily.  
Under a clear separation of time-scales between system and bath---often termed the Markov assumption---the same time-local (i.e., memoryless) generator that governs the evolution of the reduced density operator also propagates those correlations, a fact encapsulated in the quantum regression theorem (QRT).  Standard expositions may be found in the monographs~\cite{gardiner2004quantum,carmichael2013statistical,breuer2002theory}.

Quantum computers resurrect the prospect of exploiting LRT and QRT in regimes that are classically inaccessible: contemporary Hamiltonian- and Lindbladian-simulation algorithms implement the required evolutions with polylogarithmic cost in the Hilbert-space dimension.  Coupled with efficient read-out protocols, they open the door to an exponential speed-up for linear-response calculations.

\subsection{Quantum regression }\label{sec:QRT}

Consider a weak, time-dependent perturbation $\delta(t)A_2$ added to $H_{\mathrm{tot}}$, with $\abs{\delta(t)}\ll 1$ being a weak disturbance, and denote by $\tilde U(t)$ the corresponding unitary evolution. Let $A_1$ be an observable that one wants to predict. 
To first order in $\delta$, one has \cite{kubo1957}
\[
\langle A_1(t)\rangle \;=\; \langle A_1(t)\rangle_0
        + \int_0^{t}\!\chi(t,t_1)\,\delta(t_1)\,dt_1,\qquad
\chi(t_1,t_2)= -i\bigl\langle\,[A_1(t_1),A_2(t_2)]\bigr\rangle_0 ,
\]
where the expectation on the left hand side is with respect to the perturbed evolution $\tilde U(t)$,  while $\langle\,\cdot\,\rangle_0$ denotes averaging with respect to the unperturbed unitary $U(t)=e^{-itH_{\mathrm{tot}}}$. A concrete example is optical conductivity, where $\delta(t)$ is a weak classical electric field,  $A_2$ is a dipole operator, and $A_1$ is the induced current.

If both observables act solely on the system, i.e.,
\[
A_1 = O_1\otimes I_B,\qquad A_2 = O_2\otimes I_B,
\]
the response kernel becomes
\begin{equation}\label{eq:chi_sys}
\chi(t_1,t_2)= -i\Bigl(\langle O_1(t_1)O_2(t_2)\rangle
                      -\langle O_2(t_2)O_1(t_1)\rangle\Bigr),
\qquad 0\le t_2\le t_1 ,
\end{equation}
with averages taken in the Heisenberg picture under $H_{\mathrm{tot}}$ and initial state $\rho_S(0)\otimes\rho_B$.

Tracing over the bath and adopting the Born-Markov (factorization) ansatz at the earlier time,
\begin{equation}\label{separable-form}
\rho_{\mathrm{tot}}(t_2)\;\approx\; \rho_S(t_2)\otimes\rho_B
\;=\; \bigl(e^{t_2\mathcal L}(\rho_S(0))\bigr)\otimes\rho_B,
\end{equation}
which yields the quantum regression theorem (QRT) \cite{lax1963,carmichael2013statistical,gardiner2004quantum,breuer2002theory}
\begin{equation}\label{eq:QRT}
\chi(t_1,t_2)= -i\,\mathrm{tr}_S\!\left(
     O_1\,e^{(t_1-t_2)\mathcal L}\bigl(\bigl[O_2,e^{t_2\mathcal L}(\rho_S(0))\bigr]\bigr)\right),
\end{equation}
where $\mathcal L$ is the (time-local) Lindblad generator for the reduced dynamics. 

The response kernel used here is defined by perturbing the \emph{microscopic} total Hamiltonian and only then tracing out the bath. In general, this differs from the susceptibility obtained by perturbing a reduced equation of motion (e.g., a Lindblad or TCL master equation). The two coincide only under very restrictive conditions, outside which the reduction and linearisation need not commute, and additional terms arise in the reduced description.

\paragraph{Limitations of the standard QRT.}
Equation \eqref{eq:QRT} brings two practical difficulties:

\begin{enumerate}[label=(\roman*)]
\item The commutator $[O_2,e^{t_2\mathcal L}\rho_S(0)]$ is not a positive operator, so it cannot be evolved directly with standard Lindblad-simulation primitives.
\item The separability (Markov) assumption \eqref{separable-form} in general fails in the absence of sharp scale separation, all of which introduce non-Markovian memory effects \cite{kurt2020non,ali2015non}.
\end{enumerate}

Several non-Markovian extensions of the QRT have been proposed \cite{goan2011non,jin2016non,ban2017linear,ban2018two,ban2019two}, but the resulting formulas typically involve integro-differential memory kernels or explicit bath operators that are awkward to encode in quantum circuits.

\subsection{Main contributions}\label{sec:contrib}

The present work develops a framework that retains the algorithmic advantages of time-local generators while remaining valid beyond the Markov limit.  Our specific contributions are:

\begin{enumerate}[label=\arabic*.]
\item \emph{A time-local non-Markovian regression formula.}  
Using a cumulant expansion, we derive a hierarchy of time-ordered terms for $\chi(t_1,t_2)$, each expressed entirely in system operators and time-local memory super-operators.  The leading term reproduces \eqref{eq:QRT}; higher orders match known strong-coupling corrections \cite{goan2010non,ban2017double}.
\item \emph{Time evolution of commutators and anti-commutators.}  
We show how every term, which involves either the time evolution of commutators or anti-commutators, in the generalized regression theorem, can be prepared and simulated efficiently.  
\item \emph{A logarithmic-scale quantum algorithm.}  
Combining the block-encodings with oblivious amplitude amplification yields an estimator for two-time correlations whose complexity is logarithmic in $\dim\mathcal H_S$ and polynomial in $1/\varepsilon$.
\end{enumerate}

To go beyond Born's assumption \eqref{separable-form}, we analyze the correlation \cref{eq:chi_sys} in the framework of open quantum systems used to derive generalized quantum master equations (GQME) \cite{breuer2002theory,hall2014canonical}. On the other hand,  rather than tracking the density operator, we show that the response function requires the dynamics of commutators and anti-commutators,   As a result, we remove the scale separation assumption, and derive the QRT that is also valid in the non-Markovian regimes. The first main result is summarised below. 
\begin{theorem}
  \label{thm:chit1t2}
  There exist time-local generators $\cL_A(t)$, $\cL_B(t)$, $\cL_C(t)$, and Lamb shift $ H_B(t)$, of order $O(\lambda^2)$,  such that
    the Kubo's response function in \cref{eq:chi_sys} can be expressed as,
       \begin{equation}
    \begin{aligned}
         \label{eq:thmchit1t2}
                 & \chi(t_1,t_2) =\\
                 & i \tr\left(O_1   e^{\tau \mathcal L_0} \big[\rho_S(t_2), O_2 \big] \right)
                 +i  \tr\left(O_1  \big( \cT e^{ \int_0^\tau \cL_C(t)dt } - e^{\tau \cL_0}  \big) \big[\rho_S(t_2), O_2\big] \right)\\
                 + & i \tr\left(O_1 e^{\tau \mathcal L_0}   \left[ \bigl(\cT e^{ \int_0^{t_2} \cL_C(t)dt } - e^{t_2 \cL_0} \bigr)\rho_S(0),  O_2  \right] \right)  \\
                 + & i \tr\left(O_1 \big( \cT e^{ \int_0^\tau \cL_0 + \cL_A(t)dt } -   e^{ \tau \cL_0} \big) [\rho_S(t_2), O_2]\right)
                  + i\tr\left( O_1  \left[ \cT \big( e^{ \int_0^\tau \cL_0 + \cL_A(t)dt} -   e^{ \tau \cL_0} \big) \rho_S(t_2), O_2\right]\right) \\
 - i &  \tr\left(O_1   \left[ \rho_S(t_2), \big( \cT e^{ \int_0^\tau \cL_0 + \cL_A(t)dt} -   e^{ \tau \cL_0} \big)^\dag O_2 \right]  \right) - \tr\left(O_1 \big( \cT e^{ \int_0^\tau \cL_0 + \cL_B(t)dt} -   e^{ \tau \cL_0} \big) \{\rho_S(t_2), O_2\}\right) \\ 
 + & \tr\left(O_1 \left\{ \big( \cT e^{ \int_0^\tau \cL_0 + \cL_B(t)dt} -   e^{ \tau \cL_0} \big) \rho_S(t_2), O_2\right\} \right) - \tr\left(O_1 \left\{ \rho_S(t_2), \cT \big( e^{ \int_0^\tau \cL_0 + \cL_B(t)dt} -   e^{ \tau \cL_0} \big)^\dag O_2 \right\} \right) \\
 + & i  \tr\left(O_1 \left[ \big( \cT e^{ \int_0^\tau \cL_0 + L_{H_B(t)}dt} -   e^{ \tau \cL_0} \big) \rho_S(t_2),    O_2\right] \right) \\-& i \tr\left(O_1 \left[\rho_S(t_2),  \big(\cT e^{ \int_0^\tau \cL_0 + L_{H_B(t)} dt} -   e^{ \tau \cL_0} \big) O_2\right] \right)
 + \mathcal{O}(\lambda^3).
    \end{aligned}
    \end{equation}
    Here $\tau=t_1-t_2\geq 0,$ 
    \begin{equation}
        \mathcal L_0 = -i[H_S, \bullet], \; \rho_S(t)= e^{t \cL_0} \rho_S(0),
    \end{equation}
correspond to free evolution,  and $\lambda $ is the system/bath coupling constant. Moreover, the time-local operators take the form of 
\begin{equation}
    \begin{aligned}
\mathcal{L}_B(t)(X) &= \sum_{j,k} b_{jk}(t)\bigl(V_k X V_j^{\dagger} -\frac12  V_j^{\dagger}V_k X - \frac12 X V_j^{\dagger}V_k\bigr),\\
\mathcal{L}_B(t)^{\dagger}(X) &= \sum_{j,k} b_{jk}(t)\bigl(V_j X V_k^{\dagger} - \frac12 X V_j V_k^{\dagger} - \frac12 V_j V_k^{\dagger} X\bigr),\\
H_B(t)&=\sum_{j,k} b_{jk}(t)V_j V_k^{\dagger},
\end{aligned}
\end{equation}
where $V_j$ is a fixed basis in $\mathcal H_S$ and $B$ refers to the Hermitian matrix with elements $[b_{j,k}]$, which only depends on the system operators $S_j$ and the two-point bath correlation function,
\[
 C_{j,k}(t)= \tr_B(C_j(t) C_k(0) \rho_B). 
\]
Also, $\mathcal{L}_A, \mathcal{L}_C$ can be determined from the system operators $H_S$ and $S_j$ and bath correlation function
\end{theorem}

As in the standard QRT~\eqref{eq:QRT}, the generalized formula mixes the evolution of density operators with that of commutators; here, however, the propagators are time-local and can be non-Markovian, which is universal \cite{carmichael2013statistical}.  We demonstrate that these evolutions can be decomposed into time-dependent Lindblad segments while preserving $\mathcal{O}(\lambda^3)$ accuracy, making them amenable to existing quantum algorithms \cite{HLL+24}.  A more subtle issue, however, is the time evolution of the commutator and anti-commutators. We address this issue with the following theorem.

\begin{restatable}{theorem}{evolvecommutator}
  \label{lemma:o2rhoo1}
  Let $\mathcal{L}$ be a Lindbladian, and $O_1$, $O_2$ be two observables given by $(\alpha, b, \epsilon)$-block-encodings $U_{O_1}$, $U_{O_2}$ with circuit complexities $c_{O_1}$ and $c_{O_2}$, respectively. Let $\mathcal{L}$ be a Lindbladian whose simulation cost is $c_{\mathcal{L}}$. Then, for any state $\rho$, whose purification can be prepared with circuit cost $c_{\rho}$, there exists a quantum algorithm that estimates $\tr(O_2 e^{\mathcal{L}t}(\{\rho, O_1\}))$, as well as $\tr(O_2 e^{\mathcal{L}t}([O_1, \rho]))$ up to additive error $\epsilon$ and success probability at least $2/3$. This quantum algorithm uses $\mathcal{O}(1)$ applications to the block-encoding of $O_1$ and $O_2$ respectively, and the evolution superoperator $e^{\mathcal{L}t}$. The gate complexity is
  \begin{align}
    \mathcal{O}((c_{O_1} + c_{O_2} + c_{\rho} + c_{\mathcal{L}})\alpha/\epsilon^{1.25}).
  \end{align}
\end{restatable}

With these new simulation tools, we can break down the estimation of the response function $\chi(t_1,t_2)$ into modular subroutines; Table~\ref{tab:quantum_algorithms} summarizes the relevant algorithms and their costs.

\begin{table}[htbp]
\centering
\caption{Quantum algorithms for the computation of the response function and their complexity}
\label{tab:quantum_algorithms}
\begin{tabular}{lccc}
\hline\hline
Term & Quantum Algorithm & Complexity & Reference\\
\hline
$e^{t\cL_0}$  &  Hamiltonian Simulation  &  $\widetilde{\mathcal{O}}(\norm{\cL_0 } t )$ & \cite{GLSW19} \\
$\cT e^{\int_0^t \cL(s) ds}$ & Time-dependent Lindblad Simulations  & $\widetilde{\mathcal{O}}(\norm{\cL }_{L^1} t )$  & \cite{HLL+24}  \\
$\tr\left(\cdot \right) $& Block-encoding+Purification+Amplitude Amplification &$\widetilde{\mathcal{O}}(1/\epsilon)$ & \cite{Rall20} \\
\hline\hline
\end{tabular}
\end{table}

Combining these results we arrive at our main result as follows.
\begin{restatable}{theorem}{algthm}
  There exists a quantum algorithm that estimates $\chi(t_1,t_2)$ in \cref{eq:thmchit1t2} up to additive error $\epsilon$ with success probability at least $2/3$ using 
  \begin{align}
    \tilde{\mathcal{O}}(T\norm{\mathcal{L}}/\epsilon^{1.25})
  \end{align}
  applications to the block-encodings of $O_1$, $O_2$, and the circuit for preparing the initial state.
\end{restatable}

Overall, these results build an explicit bridge between modern quantum algorithms and long-standing problems in non-equilibrium quantum dynamics, providing a scalable route to compute two-point correlation functions in regimes where non-Markovian effects cannot be ignored.

\subsection{Related works}

\paragraph{Quantum algorithms for time correlations and Green’s functions.}
Algorithms for estimating real‐time correlation functions in \emph{closed} systems
have been proposed, e.g.\ the Hadamard-test protocol of
Pedernales \textit{et al.}\ \cite{pedernales2014efficient} and the block-encoding approach of Rall \cite{Rall20}.
Kökcü \textit{et al.}\ \cite{kokcu2023linear} analysed how such correlators feed directly into
linear-response functions.
All of these methods rely exclusively on unitary (Hamiltonian) evolution
and therefore do not address open-system, non-Markovian settings. 

\paragraph{Classical references for QRT.}
Foundational treatments of QRT and its Markovian
assumptions are collected in the textbooks of
Gardiner and Zoller, Carmichael, and Breuer-Petruccione
\cite{gardiner2004quantum,carmichael2013statistical,breuer2002theory}.
Our Theorem \ref{thm:chit1t2} may be viewed as a systematic extension of
those classic results to the non-Markovian, strong-coupling regime.

\paragraph{Master-equation approaches for generalized QRT.}
Ban and co-workers developed a systematic extension of linear-response theory and two-time correlations for open quantum systems \cite{ban2017linear,ban2017double,ban2018two,ban2019two}. In \cite{ban2017linear}, the external classical field is incorporated at the level of the reduced dynamics via the projection-operator formalism, yielding both time-nonlocal (Nakajima-Zwanzig) and time-local (time-convolutionless) master equations; the resulting response function explicitly separates contributions from initial system-reservoir correlations. In \cite{ban2017double}, the authors introduced double-time correlation functions of two quantum operations-encompassing ordinary correlators, linear response, and weak values-and derive their evolution with both time-convolution and time-convolutionless techniques. For Gaussian reservoirs, \cite{ban2018two} provides an exact representation of two-time correlators in terms of functional derivatives with respect to fictitious source fields, together with a perturbative expansion controlled by the reservoir correlation time, while \cite{ban2019two} applies these formulas to a two-level system to quantify “quantumness’’ via sequential measurement statistics.

\paragraph{Other generalizations of QRT.}
Goan and co-workers derived finite-temperature two-time correlation functions for open systems beyond the quantum regression theorem in two complementary settings.  For exactly solvable pure-dephasing models, \cite{goan2010non} gives closed analytic expressions that explicitly demonstrate the breakdown of QRT.  For general system-bath couplings at weak coupling, \cite{goan2011non} obtains evolution equations for two-time correlators using a non-Markovian master-equation expansion that retains bath-memory effects through explicit reservoir correlation functions and time-convolution terms.  These formulas are well suited for analytical and numerical studies of non-Markovianity, but they remain time-nonlocal and involve bath operators or kernels directly. 

\medskip 

By contrast, our construction starts from the global two-time correlator and produces a \emph{system-only}, time-local hierarchy that preserves Lindblad structure at each order; this makes the resulting expressions directly compatible with block-encoding and gate-based primitives without introducing memory-kernel convolutions.

\bigskip

The rest of the paper is organized as follows: \cref{sec: deri} derives our central result, a generalized non-Markovian regression formula (Theorem 1.1), using a systematic perturbative expansion. \cref{sec: alg} then introduces the quantum algorithms required to implement this formula, including novel methods for preparing commutators and anti-commutators, and elaborations of the analysis of the overall algorithm's complexity (Theorem 1.3). Further discussions are presented in \cref{sec: sum}.

\section{The derivation of the non-Markovian response function} \label{sec: deri}

Our primary focus is the estimate the response function to leverage the linear response theory to predict the dynamics of observables in non-equilibrium quantum systems. The two-time correlation involves two observables that are both system operators,  
\[
  A_1= O_1 \otimes I_B, \quad  A_2= O_2 \otimes I_B.
\]

Due to the continuous interactions with the bath $(B)$, the mathematical formulations of 
the dynamics of the system 
must start with the total density matrix $\rho_\text{tot}$ including the environment follows the von Neumann equation 
\begin{equation}\label{lvn}
    \partial_t\rho_\text{tot}(t) = -i[H_\text{tot}, \rho_\text{tot}], \quad H_\text{tot} := H_S \otimes I_B + I_S \otimes H_B +  H_{SB}
\end{equation}
where the dynamics is placed in a product Hilbert space $\mathcal{H}_S \otimes \mathcal{H}_B$, and the total Hamiltonian contains the system, bath, and interaction part, respectively \cite{breuer2002theory}. We follow standard settings in the literature. Specifically
\begin{enumerate}
    \item {\bf Factorized initial state: }
there is no initial entanglement, 
\begin{equation}\label{rhotot0}
    \rho_\text{tot}(0) = \rho_S(0)\otimes \rho_B,
\end{equation}
\item {\bf Stationary bath: } the initial bath density is invariant under $H_B$: $ [H_B, \rho_B]=0.$
\item {\bf Weak coupling: } $\lambda \ll 1.$ 
\item {\bf System bath coupling: }  the interaction term often takes the form of $H_{SB} = \sum_{j = 1}^J S_j\otimes B_j. $  Without loss of generality, one can assume that \cite{carmichael2013statistical},
\begin{equation}\label{trb}
    \tr(B_j \rho_B) =0, \; \forall j.
\end{equation}
\end{enumerate}   

We are focused on studying the dynamical properties of the system's observables, with a particular emphasis on estimating expectations and two-point correlations. The correlations embody the response properties of a quantum system. 

\subsection{First-order statistics}
Let $O$ be a system observable, then the expectation is given by,
\begin{equation}
     \langle O(t) \rangle := \tr\left((O\otimes I_B) \rho_\text{tot}(t) \right).
\end{equation}
We can compute the expectation by first tracing out the bath degrees of freedom,
\begin{equation}\label{eq: ot}
    \langle O(t) \rangle = \tr_S \Big(O \tr_B\left(U(t) \rho_S(0)\otimes \rho_B U(t)^\dag  \right) \Big).
\end{equation}

The formula in the above equation involves a partial trace over the bath, of a unitary dynamics with a separable initial state. This is at the heart of the theory of open quantum systems, and the reduction has been extensively studied in OQS \cite{breuer2002theory}. We summarize a result regarding a $O(\lambda^3)$ approximation as follows, 
\begin{theorem}
    There is a time-local Lindblad equation, $\frac{d}{dt}\rho_S = \cL_{TL} (t)\rho_S, $ such that,
    \begin{equation}
        \langle O(t) \rangle = \displaystyle \tr_S \Big(O \cT e^{ \int_0^t \cL_{TL}(s) ds } \rho_S(0) \Big) + O(\lambda^3 t^3).
    \end{equation}
    
\end{theorem}

\medskip

Under a weak coupling assumption, i.e., $\lambda = \norm{H_{SB}}\ll 1,$ one can apply a perturbation expansion to \eqref{lvn}, and obtain 
\begin{equation}\label{lim-lindblad}
\begin{aligned}
        & \tr_B\left( U(t) \rho_S(0) \otimes \rho_B U(t)^\dag \right) \\
        \approx&  \rho_S(t) \otimes \rho_B   -  \sum_j \sum_k \int_0^t \int_0^{t_1} \tr_B \Big(
        S_j (t_1 - t)  S_k(t_2 -t) \rho_S(t)     \otimes B_j(t_1 -t) B_k(t_2-t) \rho_B \\
      &\qquad  \qquad -  S_j (t_1 - t) \rho_S(t) S_k(t_2 -t)  \otimes B_j(t_1 -t) \rho_B B_k(t_2-t)   \\
      &\qquad \qquad  -  S_k(t_2 -t) \rho_S(t)  S_j (t_1 - t)     \otimes  B_k(t_2-t) \rho_B B_j(t_1 -t) \\
      &\qquad \qquad  +  \rho_S(t) S_k(t_2 -t)   S_j (t_1 - t)     \otimes \rho_B B_k(t_2-t)  B_j(t_1 -t) \Big) dt_2 dt_1 + O(\lambda^4).
\end{aligned}
\end{equation}
Here $\rho_S(t)= e^{t\cL_0} \rho_S(0).$

After taking the partial trace, the bath correlation functions emerge.
Specifically, we define the two-point bath correlation functions,
\begin{equation}\label{bcf}
    C_{j,k}(t_1-t_2) = \tr_B(\rho_B B_j(t_1) B_k(t_2)).
\end{equation}

Now, after tracing out the bath, we arrive at,
\begin{equation}\label{eq: non-mark}
    \tr_B\left( U(t) \rho_S(0) \otimes \rho_B U(t)^\dag \right)  = e^{t\cL_0}\rho_S(0) + \cM_2(\rho_S(t), t) + O(\lambda^3),
\end{equation}
where $\cM_2$ is known as the cumulant \cite{breuer2002theory}
    \begin{equation}\label{eqn:cM2}
    \begin{aligned}
      &  \cM_2[\rho_S(t),t]   \\
       =& -  \sum_j \sum_k \int_0^t \int_0^{t_1}  \Big(
        S_j (t_1 - t)  S_k(t_2 -t) \rho_S(t)      C_{jk}(t_1-t_2)  
   -  S_j (t_1 - t) \rho_S(t) S_k(t_2 -t)   C_{kj}(t_2-t_1)    \\
      & -  S_k(t_2 -t) \rho_S(t)  S_j (t_1 - t)      C_{jk}(t_1-t_2)    +  \rho_S(t) S_k(t_2 -t)   S_j (t_1 - t)      C_{kj}(t_2-t_1)  \Big) dt_2 dt_1.
    \end{aligned}
\end{equation}
Similarly, higher-order cumulants are multiple integrals that involve higher-order bath correlation functions. Since the two-point BCF \eqref{bcf} are the most commonly used, we will work with the approximation from $\cM_2$.  

The dynamics in \eqref{eq: non-mark} can be reduced to Lindblad, if the bath correlation length is close to zero, i.e., there is a scale separation \cite{carmichael2013statistical}. In general, the dynamics in  \eqref{eq: non-mark}, however, is non-Markovian. One canonical representation of non-Markovian quantum dynamics is through the time-local quantum master equations \cite{hall2014canonical}, 
 \begin{equation}\label{TLQ}
\frac{d}{dt} \rho_S= 
 - i[H_S + \Delta H(t),\rho_S ] + \sum_{j,k=1}^{N-1} c_{jk}(t)\bigl(2V_k \rho_S V_j^{\dagger} - V_j^{\dagger}V_k \rho_S - \rho_S V_j^{\dagger}V_k\bigr).
 \end{equation} 
Here the time-dependent
operators are projected to $V_j$'s, which are usually chosen to be fixed orthonormal basis in $\mathcal H_S$ so that time-dependence is moved to a set of coefficients $c_{jk}(t)$. 
 The generator on the right hand side, denoted by $\cL_{C}(t)$, has a time-dependence, due to the removal of the time scale separation assumption. 
 The Hermitian matrix $\Delta H(t)$ acts as a Lamb shift. Meanwhile, the Hermitian matrix $C=(c_{j,k})$ can be related to the bath correlation functions and the projection of the system operator $S_j(t)$ on a fixed basis $V_j.$ Unlike Markovian dynamics, the matrix $C$ might not be positive definite.

\begin{remark}
  In the weak coupling setting, both $\Delta H$ and $(a_{j,k})$ are $\mathcal{O}(\lambda^2)$. Therefore, the cumulant expansion in \cref{eq: non-mark}
    can be viewed as the second-order Dyson series expansion of \cref{TLQ}. It will be convenient for later derivations to express the cumulant as,
    \begin{equation}\label{TO2M2}
      \cM_2 ( \rho_S(t),t) =\big(\cT e^{\int_0^t \cL_{TL}(t')dt' } - e^{t \cL_0} \big) \rho_S(0) + \mathcal{O}(\lambda^3 t^3).
\end{equation}
\end{remark}

\begin{remark}
The matrix $[c_{jk}(t)]$ is Hermitian, and therefore, a unitary transformation can reduce it to a diagonal matrix, leading to a diagonal Lindblad-like equation. 
Meanwhile,  for non-Markovian dynamics, the matrix $[c_{jk}(t)]$ may not be positive definite. As a result, it can not be directly implemented by a direct Lindblad simulation algorithm due to the lack of CPTP property.  However, 
we can separate $c_{j,k}= a_{j,k } - b_{j,k}$ with both 
since $c_{jk}(t) = O(\lambda^2),$ one can separate $[a_{jk}(t)]$ and $[b_{jk}(t)]$ being positive definite. This separates $\cM_2$ accordingly, 
$\cM_2 = \cM_2^+ - \cM_2^-$, both of which, as in \cref{TO2M2}, by be estimated by simulating Lindblad dynamics. 
 \end{remark}

\bigskip 

\subsection{Two-point correlations. }  We now move to the two-point correlation functions in \cref{eq:chi_sys}.

As elaborated in \cite{carmichael2013statistical}, such correlation functions can be simply obtained from a Heisenberg picture for the observables. Further inspired by the treatment in \cite{carmichael2013statistical}, we simplify \cref{eq:chi_sys} to, 
\begin{equation}\label{eq: chit1t2}
     \chi(t_1,t_2)=-i \langle O_1(t_1) O_2(t_2) \rangle + i \langle  O_2(t_2) O_1(t_1) \rangle.
\end{equation}
In light of the time ordering $t_1 \geq t_2 \geq 0$, we can write the second term as 
\begin{equation}\label{eq: o1t1o2t2'}
     \big\langle O_2(t_2) O_1(t_1) \big\rangle = \tr_S\Big(O_1
     \tr_B\left( 
     U(t_1-t_2) \rho_\text{tot}(t_2)  O_2\otimes I_B  U(t_1-t_2)^\dag  
 \right) \Big).
\end{equation} 
Therefore, the propagator $U(t_1-t_2)$ is forward in time. 
As a result, we arrive at a compact formula for the response function involving a commutator for $t_1 \geq t_2$,
\begin{equation}\label{chi-comm}
     \chi(t_1,t_2)=-i \tr_S\Big(O_1
    {\tr_B\left( 
     U(t_1-t_2) \left[O_2\otimes I_B, \rho_\text{tot}(t_2)   \right]  U(t_1-t_2)^\dag  
 \right)} \Big).
\end{equation}
We show the derivation of this formula in \cref{proof-chi} for completeness.

The QRT \eqref{eq:QRT} provides a short-time approximation where the separable form \eqref{separable-form} can be assumed. To derive a more accurate approximation that goes beyond \eqref{separable-form} and the Markov assumption, we use a perturbative approximation, and replace the empirical form \eqref{separable-form} by a systematic expansion,
\begin{equation}\label{rho-expansion}
  \rho_\text{tot}(t) = \rho_\text{tot}^{(0)}(t) + \rho_\text{tot}^{(1)}(t) + \rho_\text{tot}^{(2)}(t) + \mathcal{O}(\lambda^3).  
\end{equation}
Here $\rho_\text{tot}^{(n)}(t) = \mathcal{O}(\lambda^n).$ Similarly, we can expand the unitary in \cref{chi-comm}  as follows,
\begin{equation}\label{U-expansion}
  U= U_0 + U_1 + U_2 + \mathcal{O}(\lambda^3),
\end{equation}
with $U_n= \mathcal{O}(\lambda^n).$

According to the theory of open quantum systems, the term $\rho^{(0)}(t)$ simply results in a system dynamics in the absence of the bath, $\rho^{(1)}(t)$ has no effect on the system dynamics, while  $\rho^{(2)}(t)$ incorporates the bath correlation functions (BCF) and leads to the dissipation and Lamb-shift terms in the Markovian regime \cite{carmichael2013statistical}. On the other hand, in the non-Markovian regime,  $\rho^{(2)}(t)$ leads to a time-local super-operator \cite{hall2014canonical}, which can also be represented through hierarchical equations.  

However, the roles of these terms in Kubo's response function \eqref{chi-comm} are quite different. Aside from the lack of \cref{separable-form}, the unitary operator $U$ in \cref{chi-comm} is being applied to a commutator. Moreover, the commutator, due to the time evolution to time $t_2$, is no longer unentangled. 

The leading term in \cref{rho-expansion} represents the dynamics without system/bath interaction: $U_0= e^{-itH_S} \otimes e^{-itH_B}$, and $\rho_\text{tot}^{(0)}(t) = e^{t\cL_0} \rho_S(0) \otimes \rho_B. $  
A substitution into \cref{chi-comm} thus yields,
\begin{equation}
\begin{aligned}
  & \tr_B\left( 
     U_0(t_1-t_2) \left[O_2\otimes I_B, \rho^{(0)}_\text{tot}(t_2)   \right]  U_0(t_1-t_2)^\dag  
 \right)\\
   = &    e^{(t_1 - t_2) \mathcal L_0} \big[e^{t_2\cL_0} \rho_S(0), O_2\big]
    +\cM_2\bigl(e^{(t_1 - t_2) \mathcal L_0}  \big[e^{t_2\cL_0}\rho_S(0), O_2\big], t_1 -t_2 \big), \\
   =& e^{(t_1 - t_2) \mathcal L_0} \big[ \rho_S(t_2), O_2\big] + 
   \big( \cT e^{ \int_0^\tau \cL_C(t)dt } - e^{\tau \cL_0}  \big) \big[\rho_S(t_2), O_2\big] + \mathcal{O}(\lambda^3).
\end{aligned}
\end{equation}
Here we have used the formula \eqref{TO2M2} derived in the previous section.  Notice that even at this step, the BCF already plays a role.  By multiplying these terms by $O_1$ and taking the trace over $\mathcal H_S$, we arrive at the first two terms in \cref{eq:thmchit1t2}. 

In addition, the second order term in \cref{rho-expansion} corresponds to the double integral on the right hand side of \cref{lim-lindblad} that leads to the second order cumulant  \cref{eqn:cM2}. 
\[
  \tr_B\left(   U(t_1-t_2) [ \rho^{(2)}_\text{tot}(t_2), O_2]\otimes I_B  U(t_1-t_2)^\dag  \right) =  e^{(t_1 - t_2) \mathcal L_0}   \left[ \bigl(\cT e^{ \int_0^{t_2} \cL_C(t)dt } - e^{t_2 \cL_0} \bigr)\rho_S(0),  O_2  \right]  + \mathcal{O}(\lambda^3). 
\]
This gives the term on the second line of \cref{eq:thmchit1t2}. 
 
The first order term $\rho_\text{tot}^{(1)}(t)$, which has partial trace zero gives highly non-trivial terms due to the commutator with $O_2\otimes I_B$. They give rise to the remaining terms in \cref{eq:thmchit1t2}.  We postpone the derivation to the appendix.  Instead, 
let us outline how each of the terms can be simulated. The first term represents the two-point correlation in the absence of the bath. The trace can be written as,
\[ \tr\left(O_1   e^{(t_1 -t_2) \mathcal L_0} \big[ e^{t_2 \cL_0} \rho_S(0), O_2 \big] \right).  \]
$e^{t \cL_0}$ corresponds to a unitary dynamics and can be simulated via Hamiltonian simulation techniques. Meanwhile, $ \big[ e^{t_2 \cL_0} \rho_S(0), O_2 \big]$ must be prepared in order to enable further Hamiltonian evolutions $ e^{(t_1 -t_2) \mathcal L_0} $. We will show in the next section that the commutator can be approximated by 
a finite difference method using Hamiltonian simulations with Hamiltonian $O_2.$

There are two terms in \cref{eq:thmchit1t2} involving the adjoint of the time evolution, e.g.,
\[
\tr_S\left(O_1   \left[ \rho_S(t_2), \big( \cT e^{ \int_0^\tau \cL_0 + \cL_A(t)dt} -   e^{ \tau \cL_0} \big)^\dag O_2 \right]  \right). 
\]
Using the cyclic property of the trace operator, we can rewrite this as,
\[
\tr_S\left(O_2 \big( \cT e^{ \int_0^\tau \cL_0 + \cL_A(t)dt} -   e^{ \tau \cL_0} \big) \big[O_1, \rho_S(t_2) \big]  \right), 
\]
which again can be simulated and estimated using the algorithms described above.

\section{Quantum Algorithms for Evolving Commutators and Anti-Commutators}
\label{sec: alg}

We first assume that $\rho$ is a density operator that we can prepare by Hamiltonian/Lindblad simulations, and $O$ is a Hermitian operator that comes from either $O_1$ or $O_2$ in \cref{eq:thmchit1t2}. We first discuss how to construct,
\[
 \big[O, \rho \big], \quad \big\{O, \rho \big\}
\]
that can be further evolved in time $t$, as 
\[
 e^{t\cL}(\big[O, \rho \big]),  e^{t\cL} (\big\{O, \rho \big\}),
\]
respectively.

The basic observation is that, 
\begin{equation}
\begin{aligned}
      \frac{d}{dt} e^{-itO} \rho e^{itO} \Big|_{t=0} & =  -i[O, \rho],   \\
       \frac{d}{dt} e^{tO} \rho e^{tO} \Big|_{t=0} & =  \{O, \rho\}.
\end{aligned}
\end{equation}
For the commutator, the derivative can be estimated by unitary dynamics that can be simulated using Hamiltonian simulation.  The anti-commutator is a bit subtle, but we will regard the transformation on the left hand side as a Kraus form of a quantum channel up to appropriate scaling.

The following lemma implements a block-encoding of $e^{-\delta O_1}$.
\begin{lemma}
  \label{lemma:exp-be}
  Let $U_{A}$ be an $(\alpha, b, \epsilon)$-block-encoding of $A$. For any $0 \leq \delta \leq 1$, a $(1, b, \epsilon)$-block-encoding of $e^{-\delta A}$ can be constructed using $\mathcal{O}(\sqrt{\max\{\alpha, \log 1/\epsilon\} \cdot \log 1/\epsilon})$ queries to $U_{A}$.
\end{lemma}

For small enough $\epsilon$, the query complexity $\mathcal{O}(\log 1/\epsilon)$ dominates.

The following lemma is an important tool for approximating the anti-commutator terms $\{O, \rho\}$ for observable $O$.
\begin{lemma}[Second--order centred difference]
\label{lemma:second-order-centred-difference}
Let $O$ be a bounded Hermitian operator on a finite-dimensional Hilbert
space, $\rho$ a density matrix, and define the one-parameter family
\[
\mathcal{M}_{\delta}(\rho)=e^{\delta O} \rho e^{\delta O},\qquad 
\delta\in[-\delta_{0},\delta_{0}],
\]
for some $\delta_{0}>0$.  
For $|\delta|\le\delta_{0}$, the following holds 
\begin{equation}
    \mathcal{M}_{\delta}(\rho)-\mathcal{M}_{-\delta}(\rho) = \delta \{O,\rho\}\;+\;R_{2}(\delta)
\end{equation}
where the remainder satisfies the operator-norm bound
\[
\;
\|R_{2}(\delta)\|
\le
\frac{4}{3}
\delta^3
\|O\|^{3}\,
e^{2\delta_{0}\norm{O}},  \qquad \forall \; |\delta|\le\delta_{0}.
\]
Further, this finite difference formula can be extended to 4th order, 
\begin{equation}
  \frac{  -\mathcal{M}_{2\delta}(\rho) +8\mathcal{M}_{\delta}(\rho)-8\mathcal{M}_{-\delta}(\rho)+\mathcal{M}_{-2\delta}(\rho)}{12} =   \delta \{O,\rho\}\;+\;R_{4}(\delta),
\end{equation}
where,
\[
\|R_{4}(\delta)\|
\le
\frac{32}{30}\,
\delta^{5}\,
\|O\|^{5}\,
e^{2\delta_{0}\|O\|}. 
\]
\end{lemma}

Now, we show how to estimate each term in \cref{eq:thmchit1t2}   using the following lemma.
\evolvecommutator*
\begin{proof}
  We first consider the anti-commutator $\{O_1, \rho\}$. For small enough $\delta$, we use the maps $\mathcal{M}_{\pm\delta}$, and $\mathcal{M}_{\pm2\delta}$ as in \cref{lemma:second-order-centred-difference}.

  These maps can be implemented with precision $\epsilon$ using \cref{lemma:exp-be} with $\mathcal{O}(1/\epsilon)$ queries to $U_{O_1}$. Because of normalization, we obtained a normalized version of $\mathcal{M}_{\pm\delta}(\rho)$ and $\mathcal{M}_{\pm2\delta}(\rho)$, i.e., $\mathcal{M}_{\pm \delta}(\rho)/\tr(\mathcal{M}_{\pm \delta}(\rho))$ and $\mathcal{M}_{\pm2\delta}(\rho)/\tr(\mathcal{M}_{2\pm\delta}(\rho))$. 
  For small $\delta$, the postselection probability is lower bounded by a constant.

  Using this normalized state as the input state of $e^{\mathcal{L}t}(\cdot)$, we use \cref{lemma:expectation} to estimate the following quantities
  \begin{align}
    \xi_1 &\coloneqq \tr(O_2 e^{\mathcal{L}t}(\mathcal{M}_{\delta}(\rho)))/\tr(\mathcal{M}_{\delta}(\rho)) \\ 
    \xi_2 &\coloneqq \tr(O_2 e^{\mathcal{L}t}(\mathcal{M}_{-\delta}(\rho)))/\tr(\mathcal{M}_{-\delta}(\rho)) \\
    \xi_3 &\coloneqq \tr(O_2 e^{\mathcal{L}t}(\mathcal{M}_{2\delta}(\rho)))/\tr(\mathcal{M}_{2\delta}(\rho)), \text{ and }\\
    \xi_4 &\coloneqq \tr(O_2 e^{\mathcal{L}t}(\mathcal{M}_{-2\delta}(\rho)))/\tr(\mathcal{M}_{-2\delta}(\rho)) 
  \end{align}
  with additive error $\mathcal{O}(\epsilon_0)$. The cost for this estimation is
  \begin{align}
    \label{eq:cost1}
    \mathcal{O}((c_{\mathcal{L}} + c_{O_2} + c_{O_1} + c_{\rho})\alpha/\epsilon_0).
  \end{align}

  To obtain an estimate of the desired quantity, we also need to use \cref{lemma:expectation} again to estimate the traces 
  \begin{align}
    \xi_5 \coloneqq \tr(\mathcal{M}_{\delta}(\rho)), \xi_6 \coloneqq \tr(\mathcal{M}_{-\delta}(\rho)), \xi_7 \coloneqq \tr(\mathcal{M}_{2\delta}(\rho)), \text{ and } \xi_8 \coloneqq \tr(\mathcal{M}_{-2\delta}(\rho))
  \end{align}
  with additive error $\mathcal{O}(\epsilon_0)$. 
  The costs for estimating $\xi_5$, $\xi_7$, $\xi_7$, and $\xi_8$ are dominated by \cref{eq:cost1}. Finally, we have
  \begin{align}
    \tr(O_2 e^{\mathcal{L}t}(\{\rho, O_1\})) = (\xi_4 \xi_5 +8\xi_1\xi_5 -8\xi_2\xi_6 + xi_4\xi_8)/(12\delta) + \mathcal{O}((\epsilon_0+\delta^5)/\delta).
  \end{align}

  For the commutator, we use a $\delta$-time Hamiltonian evolution to approximate it. In particular, we use a variation of \cref{lemma:second-order-centred-difference} with the Kraus operator for $\mathcal{M}_{\delta}$ replaced by $e^{i\delta O}$. Similar to the above analysis for the anti-commutator terms, the commutator terms $\tr(O_2 e^{\mathcal{L}t}([\rho, O_1]))$ can be estimated with error $\mathcal{O}((\epsilon_0 + \delta^5)/\delta)$ with the cost the same as \cref{eq:cost1}.

It suffices to choose $\epsilon_0 = \mathcal{O}(\epsilon^{5/4})$ and $\delta = \mathcal{O}(\epsilon^{1/4})$, to bound the estimate error by $\epsilon$.

\end{proof}

\algthm*
\begin{proof}
  This algorithm is based on \cref{thm:chit1t2} and \cref{lemma:o2rhoo1}. 
  We first label each term of \cref{eq:thmchit1t2} as follows.
  \begin{equation}
      \begin{aligned}
           \label{eq:thmchit1t2-labeled}
                   & \chi(t_1,t_2) =\\
                   & \underbrace{i \tr\left(O_1   e^{\tau \mathcal L_0} \big[\rho_S(t_2), O_2 \big] \right)}_{\textcircled{\raisebox{-.9pt} {1}}}
                   +\underbrace{i  \tr\left(O_1  \big( \cT e^{ \int_0^\tau \cL_C(t)dt } - e^{\tau \cL_0}  \big) \big[\rho_S(t_2), O_2\big] \right)}_{\textcircled{\raisebox{-.9pt} {2}}}\\
        + & \underbrace{i \tr\left(O_1 e^{\tau \mathcal L_0}   \left[ \bigl(\cT e^{ \int_0^{t_2} \cL_C(t)dt } - e^{t_2 \cL_0} \bigr)\rho_S(0),  O_2  \right] \right)}_{\textcircled{\raisebox{-.9pt} {3}}}  \\
        + & \underbrace{i \tr\left(O_1 \big( \cT e^{ \int_0^\tau \cL_0 + \cL_A(t)dt } -   e^{ \tau \cL_0} \big) [\rho_S(t_2), O_2]\right)}_{\textcircled{\raisebox{-.9pt} {4}}}
        + \underbrace{i\tr\left( O_1  \left[ \cT \big( e^{ \int_0^\tau \cL_0 + \cL_A(t)dt} -   e^{ \tau \cL_0} \big) \rho_S(t_2), O_2\right]\right)}_{\textcircled{\raisebox{-.9pt} {5}}} \\
        - &\underbrace{i   \tr\left(O_1   \left[ \rho_S(t_2), \big( \cT e^{ \int_0^\tau \cL_0 + \cL_A(t)dt} -   e^{ \tau \cL_0} \big)^\dag O_2 \right]  \right)}_{\textcircled{\raisebox{-.9pt} {6}}} - \underbrace{\tr\left(O_1 \big( \cT e^{ \int_0^\tau \cL_0 + \cL_B(t)dt} -   e^{ \tau \cL_0} \big) \{\rho_S(t_2), O_2\}\right)}_{\textcircled{\raisebox{-.9pt} {7}}} \\ 
        + & \underbrace{\tr\left(O_1 \left\{ \big( \cT e^{ \int_0^\tau \cL_0 + \cL_B(t)dt} -   e^{ \tau \cL_0} \big) \rho_S(t_2), O_2\right\} \right)}_{\textcircled{\raisebox{-.9pt} {8}}} - \underbrace{\tr\left(O_1 \left\{ \rho_S(t_2), \cT \big( e^{ \int_0^\tau \cL_0 + \cL_B(t)dt} -   e^{ \tau \cL_0} \big)^\dag O_2 \right\} \right)}_{\textcircled{\raisebox{-.9pt} {9}}} \\
        + & \underbrace{i  \tr\left(O_1 \left[ \big( \cT e^{ \int_0^\tau \cL_0 + L_{H_B(t)}dt} -   e^{ \tau \cL_0} \big) \rho_S(t_2),    O_2\right] \right)}_{\textcircled{\raisebox{-.9pt} {10}}} \\
        -& \underbrace{i \tr\left(O_1 \left[\rho_S(t_2),  \big(\cT e^{ \int_0^\tau \cL_0 + L_{H_B(t)} dt} -   e^{ \tau \cL_0} \big) O_2\right] \right)}_{\textcircled{\raisebox{-.9pt} {11}}}
  + \mathcal{O}(\lambda^3).
      \end{aligned}
  \end{equation}
  The terms \textcircled{\raisebox{-.9pt} {1}}, \textcircled{\raisebox{-.9pt} {2}}, \textcircled{\raisebox{-.9pt} {3}}, \textcircled{\raisebox{-.9pt} {4}}, and \textcircled{\raisebox{-.9pt} {7}} can be estimated using \cref{lemma:o2rhoo1}. The terms \textcircled{\raisebox{-.9pt} {5}}, \textcircled{\raisebox{-.9pt} {8}}, and \textcircled{\raisebox{-.9pt} {10}} are more straightforward: we just apply time-dependent Lindbladian simulation and then \cref{lemma:expectation}. Estimating term \textcircled{\raisebox{-.9pt} {6}} and \textcircled{\raisebox{-.9pt} {9}} can be reduced to the first case, as for a map $\mathcal{M}$, state $\rho$, and observables $O_1, O_2$, we have
  \begin{align}
    \tr(O_1[\rho, \mathcal{M}^{\dag}(O_2)]) &= \tr(O_1\rho\mathcal{M}^{\dag}(O_2)) - \tr(O_1\mathcal{M}^{\dag}(O_2)\rho) \\
                                            &= \tr(O_1\rho\mathcal{M}^{\dag}(O_2)) - \tr(\rho O_1\mathcal{M}^{\dag}(O_2)) \\
                                            &=\tr(O_2\mathcal{M}(O_1\rho)) - \tr(O_2\mathcal{M}(\rho O_1)) \\
                                            \label{eq:o2mo1rho}
                                            &=\tr(O_2\mathcal{M}([O_1, \rho])).
  \end{align}
  The term \textcircled{\raisebox{-.9pt} {11}} can be dealt with similarly, with the exception that the map $\cT e^{ \int_0^\tau \cL_0 + L_{H_B(t)} dt} -   e^{ \tau \cL_0}$, instead of its adjoint, is applied to $O_2$. Once it is converted to the form as in \cref{eq:o2mo1rho}, it is in the form of $\tr(O_2\mathcal{M}^{\dag}[O_1, \rho])$. This adjoint map $\mathcal{M}^{\dag}$ is easy to implement because $\cT e^{ \int_0^\tau \cL_0 + L_{H_B(t)} dt} -   e^{ \tau \cL_0}$ is Hamiltonian evolution.

  The simulation costs are summarized in \cref{tab:quantum_algorithms}. Note that the Lindbladian is given in the GKS form (see \cref{TLQ}) instead of the Lindblad form. Despite the additional complication in representation, it does not increase the complexity of the simulation algorithm, as one can efficiently convert the GKS form to a Lindblad form as shown in \cite{CL17}. Then, the claimed cost follows from \cref{lemma:o2rhoo1}.

\end{proof}

\section{Summary and Discussions}\label{sec: sum}

We have extended linear-response theory to an open-system setting in which the response kernel is written entirely in terms of reduced system dynamics.  In particular, we derived a non-Markovian generalisation of the quantum regression theorem that remains \emph{time-local} and is expressed through simulation primitives available on quantum hardware (time-dependent Hamiltonian and Lindblad evolutions).  We further provided efficient routines for propagating commutators and anti-commutators, yielding a modular algorithm for estimating two-point response functions.   The resulting estimators have cost polylogarithmic in $\dim\mathcal H_S$, indicating an asymptotic advantage over classical approaches whose cost scales at least polynomially in the system dimension.

In summary, the paper establishes a path from non-Markovian linear response to concrete, resource-efficient quantum algorithms.  We expect the techniques introduced here to go beyond the typical simulation tasks in open quantum systems and provide a useful toolkit for studying non-equilibrium properties. 

\smallskip 

The are two promising generalizations that can further extend the predicative capability of the current method. 

First, removing the weak-coupling assumption is a natural next step.  
For Gaussian environments, the bath is completely characterised by its two-point correlation function. Such a direction has been pursued in \cite{ban2018two} with a perturbative approximation constructed with a diagrammatic expansion.  The nonperturbative embeddings (e.g., pseudomode ) can generate exact time-local generators for the reduced dynamics \cite{tamascelli2018nonperturbative,huang2024unified} and point to a different direction.  
Adapting our construction to such embeddings would give formulas for the response function at arbitrary coupling while preserving the system-only, time-local structure that makes the present approach algorithmically attractive.

Secondly,  extending the derivation to nonlinear response requires multi-time correlation functions.  
While higher-order Kubo formulas are well known in the Markovian setting \cite{carmichael2013statistical}, carrying out time-local, system-only expansion to $n$-point correlators would provide systematically improvable predictions beyond the separability assumption.  
The main challenges are the combinatorial growth of nested (anti-)commutators and the control of truncation errors.

\section*{Acknowledgement}
Li's research on this project has been supported by the NSF Grant No. DMS-2111221 and No. CCF-2312456. CW acknowledges support from NSF CCF-2312456 and CCF-2238766.

\appendix

\section{Block-encoding and related tools}
Let $A$ be an operator acting on $n$ qubits, we say that an $(n+b)$-qubit unitary $U_A$ is an $(\alpha, b, \epsilon)$-\emph{block-encoding} of $A$ if
\begin{align}
  \norm{A - \alpha (\bra{0^{\otimes b}} \otimes I^{2^n}) U_A (\ket{0^{\otimes b}} \otimes I_{2^n})} \leq \epsilon,
\end{align}
where $I_{2^n}$ is the identity operator acting on $n$ qubits. Intuitively, $A$ appears in the upper-left block of $A$:
\begin{align}
  U_A = 
  \begin{pmatrix}
    A/\alpha & \cdot\\
    \cdot & \cdot
  \end{pmatrix},
\end{align}
and we call $\alpha$ the normalizing factor.

The following lemma from~\cite{Rall20} is used to estimate the expectation of block-encoded observables.
\begin{lemma}[{\cite{Rall20}}]
  \label{lemma:expectation}
  Let $A$ be a Hermitian which can be block-encoded by a unitary with scaling factor $\alpha$ and implementing cost $Q$. Let $\rho$ be a state whose purification can be prepared by a circuit with cost $R$. For all $\epsilon, \delta > 0$, there exists a quantum algorithm that produces an estimate $\xi$ of $\tr(\rho A)$ such that
  \begin{align}
    |\xi - \tr(\rho A)| \leq \epsilon
  \end{align}
  with probability at least $1-\delta$. This algorithm has circuit complexity $\mathcal{O}((R+Q)\frac{\alpha}{\epsilon}\log\frac{1}{\delta})$.
    
\end{lemma}

\section{Derivation of the Compact Formula for the Response Function}\label{proof-chi}

We will derive the compact expression for the response function in \eqref{chi-comm}.

The starting point is the Kubo formula in linear response theory:
\[
\chi(t_1, t_2) = -i \langle [O_1(t_1), O_2(t_2)] \rangle = -i \operatorname{Tr}_{SB}\left( [O_1(t_1), O_2(t_2)] \, \rho_{\text{tot}}(0) \right).
\]
Here, $O_1$ and $O_2$ are Hermitian operators acting on the system Hilbert space $\mathcal{H}_S$, and are lifted to the total Hilbert space $\mathcal{H}_S \otimes \mathcal{H}_B$ by defining $A_1 = O_1 \otimes \mathbb{I}_B$ and $A_2 = O_2 \otimes \mathbb{I}_B$. The time evolution is generated by the total Hamiltonian $H_{\text{tot}}$ through the unitary operator
\[
U(t) := e^{-i H_{\text{tot}} t}.
\]
The Heisenberg picture observables are defined as:
\[
A_j(t) := U^\dagger(t) A_j U(t), \quad j = 1,2.
\]

Expanding the commutator in the Kubo formula, we write:
\[
\chi(t_1, t_2) = -i \operatorname{Tr}_{SB} \left( A_1(t_1) A_2(t_2) \rho_0 \right) + i \operatorname{Tr}_{SB} \left( A_2(t_2) A_1(t_1) \rho_0 \right),
\]
where $\rho_0 := \rho_{\text{tot}}(0)$.

We start with the first term. By using the Heisenberg picture of the observables and the Schr\"odinger picture of the density operator, and the cyclic property of trace, we have
\[
\begin{aligned}
-i \operatorname{Tr}_{SB} \left( A_1(t_1) A_2(t_2) \rho_0 \right) &= -i \operatorname{Tr}_{SB} \left( U^\dagger(t_1) A_1 U(t_1) U^\dagger(t_2) A_2 U(t_2) \rho_0 \right) \\
  &= -i \operatorname{Tr}_{SB} \left(  A_1 U(t_1) U^\dagger(t_2) A_2 U(t_2) \rho_0 U^\dagger(t_1) \right) \\
  &= -i \operatorname{Tr}_{SB} \left(  A_1 U(t_1-t_2)  A_2  \rho_0(t_2) U^\dagger(t_1-t_2) \right). 
\end{aligned}
\]

Similarly, for the second term, we have
\[
\begin{aligned}
i \operatorname{Tr}_{SB} \left( A_2(t_2) A_1(t_1) \rho_0 \right) &= i
\operatorname{Tr}_{SB} \left( A_1(t_1) \rho_0  A_2(t_2) \right) \\
&= i \operatorname{Tr}_{SB}   \left( A_1 U(t_1) \rho_0  U^\dag(t_2) A_2 U^\dag(t_1-t_2) \right) \\
&= i \operatorname{Tr}_{SB}   \left( A_1 U(t_1-t_2) \rho_S(t_2)  A_2 U^\dag(t_1-t_2) \right) \\
\end{aligned}
\]
This matches the required result in \eqref{chi-comm}.

\section{Asymptotic Analysis of the system-bath dynamics. }

We let $U_\text{tot}(t)= e^{-itH_\text{tot}}$ be the unitary operator for the exact system-bath dynamics according to \eqref{lvn}. We let $U_0(t)= e^{-it(H_S \otimes I_B + I_S \otimes H_B)}=U_S(t) \otimes U_B(t)$ be the uncoupled unitary operator that can be separated into unitary dynamics of the system and bath, i.e., $U_S$ and $U_B$, respectively.  For simplicity of the expressions, we denote for an operator $R$ pertaining to the density operator,
\begin{equation}
    \Gamma(t):=  U_S(t) \Gamma U_S(t)^\dag, \quad  R(t):=  U_B(t) R U_B(t)^\dag.
\end{equation}
For example, $\rho_S(t)$ indicates the dynamics of the quantum system when the interaction with the bath is absent, i.e., free evolution. 

Meanwhile, for the system and bath operators, we define their time evolution according to the Heisenberg picture,
\begin{equation}
    S_j(t):=  U_S(t)^\dag S_j U_S(t), \quad  S_j(t):=  U_B(t)^\dag B_j U_B(t).
\end{equation}

Let us study the dynamics with an unentangled initial state, (e.g., \eqref{rhotot0}), 
\begin{equation}
     \partial_t\rho(t) = -i[H_\text{tot}, \rho], \quad \rho_\text{tot}(0) = \Gamma \otimes R. 
\end{equation}
Here $\Gamma$ and $R$ will play the role of corrections to $\rho_S$ and $\rho_B$. Therefore, they themselves may not be positive. By the weak coupling assumption, i.e., $\norm{H_{SB}} \ll 1$, we can expand the solution of the equation above,
\begin{equation}\label{asymp}
    \rho(t) = \rho^{(0)}(t) + \rho^{(1)}(t) + \rho^{(2)}(t) + O(\norm{H_{SB}}^3), 
\end{equation}
where $\rho^{(n)}(t) = O(\norm{H_{SB}}^n)$. With a direct asymptotic expansion, we find that,
\begin{equation}\label{asymp1}
 \begin{aligned}
         \rho^{(0)}(t) = & U_S(t) \Gamma U_S(t)^\dag \otimes  U_B(t) R U_B(t)^\dag, \\
          \rho^{(1)}(t) = & -i \int_0^t \big[H_{SB} (t_1 - t), \rho^{(0)}(t) \big ] dt_1, \\
          \rho^{(2)}(t) = & -i \int_0^t \big[H_{SB} (t_1 - t), \rho^{(1)}(t) \big ] dt_1, \\ 
           \cdots &\quad  \cdots 
 \end{aligned}
\end{equation}

Equivalently, these terms can be obtained by using Duhamel's principle repeatedly,
\begin{equation}\label{U-duhamel}
    U_\text{tot}(t)= U_0(t) -i \int_0^t H_{SB}(t_1-t) U_\text{tot}(t) ] dt_1.
\end{equation}

With these notations in place, we have,
\begin{equation}\label{rho123}
\begin{aligned}
     \rho^{(0)}(t) =& \Gamma(t) \otimes R(t),\\
    \rho^{(1)}(t) =& -i \int_0^t \big[H_{SB} (t_1 - t), \Gamma(t) \otimes R(t) \big ] dt_1\\
     =&  -i \sum_j \int_0^t S_j (t_1 - t) \Gamma(t)  \otimes B_j(t_1 -t) R(t)    - \Gamma(t) S_j (t_1 - t) \otimes R(t) B_j(t_1 -t)  dt_1 \\
      \rho^{(2)}(t) =& -i \int_0^t \big[H_{SB} (t_1 - t), U_0(t-t_1)  \rho^{(1)}(t_1)  U_0(t-t_1)^\dag \big ] dt_1\\
      =& - \sum_j \sum_k \int_0^t \int_0^{t_1}  \Big(
        S_j (t_1 - t)  S_k(t_2 -t) \Gamma(t)     \otimes B_j(t_1 -t) B_k(t_2-t) R(t)  \\
      &\qquad  \qquad -  S_j (t_1 - t) \Gamma(t) S_k(t_2 -t)  \otimes B_j(t_1 -t) R(t) B_k(t_2-t)   \\
      &\qquad \qquad  -  S_k(t_2 -t) \Gamma(t)  S_j (t_1 - t)     \otimes  B_k(t_2-t) R(t) B_j(t_1 -t) \\
      &\qquad \qquad  +  \Gamma(t) S_k(t_2 -t)   S_j (t_1 - t)     \otimes R(t) B_k(t_2-t)  B_j(t_1 -t) \Big) dt_2 dt_1.\\
\end{aligned}
\end{equation}

\section{The proof of the main theorem \eqref{thm:chit1t2}}

We present a derivation of the response function 
\begin{equation}\label{eq: chi}
     \chi(t_1,t_2)= i \tr_S\Big(O_1
    {\tr_B\left( 
     U(t_1-t_2) \left[\rho_\text{tot}(t_2),  O_2\otimes I_B \right]  U(t_1-t_2)^\dag  
 \right)} \Big),
\end{equation}
where $U(t) $ the unitary dynamics of the system/bath Hamiltonian $ H_\text{tot} := H_S \otimes I_B + I_S \otimes H_B +  H_{SB}.$

 We follow the assumption \eqref{separable-form} that the system and bath are separable and the perturbation analysis. 
  Specifically, inside the equation \eqref{eq: chi}, we insert
\begin{equation}\label{rhotot012}
  \rho_\text{tot}(t_2)  =   \rho^{(0)}_\text{tot}(t_2)+  \rho^{(1)}_\text{tot}(t_2) +   \rho^{(2)}_\text{tot}(t_2)   + \mathcal{O}(\lambda^3).
\end{equation}
Similarly, the unitary operator $U= U_0 + U_1 + \cdots $. 

We separate the response function into three terms accordingly and treat them separately,
\begin{equation}\label{chi-123}
     \chi(t_1,t_2) = i \tr\left(O_1 (\i + \ii + \iii) \right).
\end{equation}

We first see that,
\begin{equation}\label{eq: term-I}
\begin{aligned}
\i = &    \tr_B\left(   U(t_1-t_2)  \big[ \rho^{(0)}_\text{tot}(t_2), O_2\otimes I_B\big]  U(t_1-t_2)^\dag \right) \\
=&  \tr_B\left(   U(t_1-t_2)   \big[ \rho_S(t_2), O_2\big] \otimes \rho_B  U(t_1-t_2 )^\dag \right) \\
=&   e^{(t_1 - t_2) \mathcal L_0} \big[ \rho_S(t_2), O_2\big]
  + \cM_2\bigl(e^{(t_1 - t_2) \mathcal L_0}  \big[\rho_S(t_2), O_2\big], t_1 -t_2 \big). 
\end{aligned}
\end{equation}
To arrive at the last line, we have used the second order cumulant expansion. In addition, we introduced $\mathcal L_0 = -i[H_S, \bullet]$.   \cref{eq: term-I} gives the first term in the main theorem. 

From \cref{TO2M2}, we can further use \eqref{TLQ} of $\cM_2$ and write the result as,
\[
\i =  e^{(t_1 - t_2) \mathcal L_0} \big[ \rho_S(t_2), O_2\big] + 
  \big( \cT e^{ \int_0^\tau \cL_C(t)dt } - e^{\tau \cL_0}  \big) \big[\rho_S(t_2), O_2\big].
\]

\bigskip 

Meanwhile, for the third term $\iii$, we have,
\begin{equation}
\begin{aligned}
& U(t_1 - t_2)\, \rho^{(2)}_{\text{tot}}(t_2)\, O_2 \otimes \mathbb{I}_B\, U^\dagger(t_1 - t_2) \\
  =&\ U_0(t_1 - t_2)\, \rho^{(2)}_{\text{tot}}(t_2)\, O_2 \otimes \mathbb{I}_B\, U_0^\dagger(t_1 - t_2) + \mathcal{O}(\lambda^3) \\
=&\ - \sum_{j,k} \int_0^{t_2} \int_0^{t_1'} 
S_j(t_1') S_k(t_2') \rho_S(0)\, O_2(t_1) \otimes B_j(t_1') B_k(t_2') \rho_B \, dt_2'\, dt_1' + \cdots + \mathcal{O}(\lambda^3),
\end{aligned}
\end{equation}
To arrive at the second line, we have used the fact that $\norm{\rho^{(2)}}=\mathcal{O}(\lambda^2)$ and $U=U_0 + \mathcal{O}(\lambda)$. On the third line, we only show the first term from $\rho^{(2)}_\text{tot}$ in \eqref{rho123}. 

Now, by taking the partial trace, we get
\begin{equation}
\begin{aligned}
   \iii= &   \tr_B\left(   U(t_1-t_2) [ \rho^{(2)}_\text{tot}(t_2), O_2]\otimes I_B  U(t_1-t_2)^\dag  \right) \\
   =  &  \left[- \sum_j \sum_k   \int_0^{t_1} \int_0^{t_1'}  \Big(    S_j (t_1' - t_2)  S_k(t_2' -t_2) \rho_S(t_1) C_{jk}(t_1' - t_2') + \cdots \Big) dt_2' dt_1',  O_2(t_2-t_1)\right]  + \mathcal{O}(\lambda^3) \\
   =  & e^{(t_1 - t_2) \mathcal L_0}   \left[ \cM_2\bigl( \rho_S(t_2), t_2\bigr),  O_2  \right]   + \mathcal{O}(\lambda^3) \\
   = &  e^{(t_1 - t_2) \mathcal L_0}   \left[ \bigl(\cT e^{ \int_0^{t_2} \cL_C(t)dt } - e^{t_2 \cL_0} \bigr)\rho_S(0),  O_2  \right] 
\end{aligned}
\end{equation}
Here we have recognized the double integral as the second order cumulant \eqref{eqn:cM2} and used  \cref{TO2M2}.

\bigskip 

It remains to estimate the contribution from \(\rho^{(1)}_{\text{tot}}\), corresponding to the second term $\ii$ in \cref{chi-123},
\[
\ii = \mathrm{Tr}_B\left( U(t_1 - t_2) \,  [\rho^{(1)}_{\text{tot}}(t_2), O_2 \otimes I_B] \, U^{\dagger}(t_1 - t_2) \right).
\]

\begin{equation}\label{rho11}
\begin{aligned}
    & U(t_2-t_1)  \rho^{(1)}_\text{tot}(t_1) O_1\otimes I_B  U(t_2-t_1)^\dag \\
     =&  U_0(t_2-t_1)  \rho^{(1)}_\text{tot}(t_1) O_1\otimes I_B  U_0(t_2-t_1)^\dag \\
      &- i \int_0^{t_2-t_1} \left[H_{SB}(t_1'-t_2+t_1), U_0(t_2-t_1) \rho^{(1)}_\text{tot}(t_1)O_1\otimes I_B  U_0(t_2-t_1)^\dag   \right] dt_1' 
\end{aligned}
\end{equation}

We have noticed that the first term on the right hand side has partial trace zero, which can be verified directly. For the second term,
recall that,
\begin{align*}
    \rho^{(1)}_\text{tot}(t) =& -i \int_0^t \big[H_{SB} (t_1 - t), \rho_S(t) \otimes \rho_B \big ] dt_1\\
     =&  -i \sum_k \int_0^t S_k (t_1 - t) \rho_S(t)  \otimes B_k(t_1 -t) \rho_B    - \rho_S(t) S_k (t_1 - t) \otimes \rho_B(t) B_k(t_1 -t)  dt_1.
\end{align*}

With a direct substitution, we obtain, the following expression for the term $\ii$:

\begin{equation}\label{ii-to-8-double-integrals}
\begin{aligned}
\ii=      &  \tr_B \left( U(t_1-t_2)  \left[ \rho^{(1)}_\text{tot}(t_2), O_2\otimes I_B\right]  U(t_1-t_2)^\dag \right)  \\
        &= - \sum_j \sum_k  \int_0^{\tau} \int_0^{t_2}  S_j( t_1'-\tau) S_k(t_2' -t_1) \rho_S(t_1)   O_2(-\tau)  C_{j,k}(t_1'-t_2' + t_2) dt_2'dt_1' \\
        &  + \sum_j \sum_k \int_0^{\tau} \int_0^{t_2} S_j( t_1'-\tau) \rho_S(t_1) S_k(t_2' -t_1)   O_2(-\tau) C_{j,k}(t_1'-t_2' + t_2)^\dag  dt_2'dt_1' \\ 
         &  + \sum_j \sum_k   \int_0^{\tau} \int_0^{t_2}  S_k(t_2' -t_1) \rho_S(t_1)   O_2(-\tau)  S_j( t_1'-\tau) C_{j,k}(t_1'-t_2' + t_2)  dt_2'dt_1' \\ 
         &- \sum_j \sum_k  \int_0^{\tau} \int_0^{t_2}    \rho_S(t_1)S_k(t_2' -t_1)  O_2(-\tau) S_j( t_1'-t_1+t_2)  C_{j,k}(t_1'-t_2' + t_2)^\dag  dt_2'dt_1'\\
         & + \sum_j \sum_k  \int_0^{\tau} \int_0^{t_2}  S_j( t_1'-\tau)  O_2(-\tau) S_k(t_2' -t_1) \rho_S(t_1)    C_{j,k}(t_1'-t_2' + t_2) dt_2'dt_1' \\
        &  - \sum_j \sum_k \int_0^{\tau} \int_0^{t_2} S_j( t_1'-\tau) O_2(-\tau) \rho_S(t_1) S_k(t_2' -t_1)    C_{j,k}(t_1'-t_2' + t_2)^\dag  dt_2'dt_1' \\ 
         &  - \sum_j \sum_k   \int_0^{\tau} \int_0^{t_2}  O_2(-\tau) S_k(t_2' -t_1) \rho_S(t_1)     S_j( t_1'-\tau) C_{j,k}(t_1'-t_2' + t_2)  dt_2'dt_1' \\ 
         &+ \sum_j \sum_k  \int_0^{\tau} \int_0^{t_2}  O_2(-\tau)  \rho_S(t_1)S_k(t_2' -t_1)   S_j( t_1'-t_1+t_2)  C_{j,k}(t_1'-t_2' + t_2)^\dag  dt_2'dt_1'.\\
\end{aligned}
\end{equation}

To simplify these terms to recognizable forms, we express the two-point bath correlation function \eqref{bcf} in a spectral form,
\begin{equation}
    C_{j,k}(t) = \sum_{\mu} g_{j,\mu} g_{k,\mu}^* e^{-it\omega_\mu}. 
\end{equation}
Here $\omega_\mu$s are the bath frequencies, and $g_{j,\mu}$ is related to the spectral density.  

To elaborate on how the integrals in \eqref{ii-to-8-double-integrals} are treated, 
we write
\[
  \ii = \sum_j \sum_k  \int_0^{\tau} \int_0^{t_2}  Q_1 + Q_2 + \cdots + Q_8  dt_2'dt_1'
\]
It is enough to consider the first integral, denoted by  $Q_1$. Other terms are simply a rearrangement of $Q_1$. 
\[
\begin{aligned}
  Q_1=&  \sum_\mu  \int_0^{\tau} \int_0^{t_2}  T_\mu(t_1' - \tau) T_\mu(t_2'- t_1) \rho_S(t_1) O_2(-\tau) e^{-i(t_1'-t_2' + t_2)\omega_\mu} dt_2'dt_1' \\
  =&   \sum_\mu   \int_0^{\tau} U(\tau-t_1') \int_0^{t_2}  T_\mu T_\mu(-t_1'+t_2'- t_2) \rho_S(t_2+t_1') O_2(-t_1')  e^{-i(t_1'-t_2' + t_2)\omega_\mu} dt_2' U(\tau-t_1')^\dag dt_1'  
\end{aligned}
\]
Here, to incorporate the spectral property, we have defined,
\begin{equation}\label{S2T}
    T_\mu = \sum_j S_j g_{j,\mu}.
\end{equation}

We now let $\{V_j\}$ be a fixed basis in $\mathcal H_S$ and we expand
\begin{equation}\label{T2V}
    T_\mu = \sum_m y_{\mu,j}(0) V_j, \quad     T_\mu(t) e^{-\omega_m t} = \sum_m y_{\mu,j}(t) V_j.
\end{equation}
As a result, we can simplify the above equation to 
\begin{equation}
\begin{aligned}
        Q_1 =& \sum_{j,k}  \int_0^{\tau} U(\tau-t_1')  d_{j,k}(t_1') V_j V_k^\dag \rho_S(t_2+t_1') O_2(-t_1') U(\tau-t_1')^\dag dt_1'  \\
         =& \sum_{j,k}  \int_0^{\tau} U(\tau-t_1')  d_{j,k}(t_1') V_j(t_1-\tau) V_k (t_1-\tau)^\dag \rho_S(t_1) O_2(-\tau)  dt_1'.   
\end{aligned}
\end{equation}
Here the coefficient matrix $D=(d_{j,k})$ is computed from the integrals of the  expansion coefficients in \eqref{T2V},
\begin{equation}\label{djk}
     d_{j,k}(t_1') = \sum_\mu  \int_0^{t_2}   y_{\mu,j}(0) y_{\mu,j}^*(-t_2'-t_1') dt_2'.
\end{equation}

This is beginning to resemble part of the solution of the time-local quantum master equation \eqref{TLQ} in Duhamel's form \eqref{TO2M2}.  
However, the coefficients here may not be associated with a Hermitian matrix. 
To simplify these formulas, we need the following calculation. We first define GKLS generator for a Hermitian matrix $A$, its adjoint, and the Lamb shift,
\begin{equation}
    \begin{aligned}
\mathcal{L}_A(X) &= \sum_{j,k} a_{jk}\bigl(V_k X V_j^{\dagger} -\frac12  V_j^{\dagger}V_k X - \frac12 X V_j^{\dagger}V_k\bigr),\\
\mathcal{L}_A^{\dagger}(X) &= \sum_{j,k} a_{jk}\bigl(V_j X V_k^{\dagger} - \frac12 X V_j V_k^{\dagger} - \frac12 V_j V_k^{\dagger} X\bigr),\\
H_A&=\sum_{j,k} a_{jk}V_j V_k^{\dagger}.
\end{aligned}
\end{equation}

We consider the following 8 terms, representing the simplified terms in \eqref{ii-to-8-double-integrals}, after the transformations \cref{S2T,T2V,djk}, but without showing the integrals and time dependence. 
\begin{equation}
    \begin{aligned}
Q= & -\sum_{j,k} d_{jk}\,V_j V_k^{\dagger}\,\rho O, 
 +\sum_{j,k} d_{jk}^*\,V_j^{\dagger}\rho V_k O, 
  +\sum_{j,k} d_{jk}\,V_k^{\dagger}\rho O V_j, 
 -\sum_{j,k} d_{jk}^*\,\rho V_k O V_j^{\dagger}, \\[4pt]
 &  +\sum_{j,k} d_{jk}\,V_j O V_k^{\dagger}\rho,
 -\sum_{j,k} d_{jk}^*\,V_j^{\dagger}O\rho V_k, 
  -\sum_{j,k} d_{jk}\,O V_k^{\dagger}\rho V_j, 
 +\sum_{j,k} d_{jk}^*\,O\rho V_k V_j^{\dagger}.
 \end{aligned}
\end{equation}

After lengthy calculations, we derived the following identity, 
\begin{equation}
    \begin{aligned}
Q=& 
  \,\mathcal{L}_{A}\!\bigl([\rho,O]\bigr)
+[\mathcal{L}_{A}\rho,O]
-[\rho,\mathcal{L}_{A}^{\dagger}O]
+i\mathcal{L}_{B}\bigl(\{\rho,O\}\bigr) 
  -i\{\mathcal{L}_{B}\rho,O\} 
  +i\{\rho,\mathcal{L}_{B}^{\dagger}O\}  \\  
&-{i}\bigl([[H_B,\rho],O] + [\rho,[H_B,O]]\bigr).
\end{aligned}
\end{equation}
Throughout, $D=(d_{jk})$ with $d_{jk}=a_{jk}+\mathrm{i}b_{jk}$, where $A=(a_{jk})$ and $B=(b_{jk})$ are Hermitian.

When each term is integrated from $0$ to $\tau$, it can be regarded as a perturbation term in \eqref{TO2M2}, and therefore, can be expressed as the difference between the time-ordered evolution of a time-local quantum master equation and a free evolution. 
We summarize the formula here. 
\begin{equation}
\begin{aligned}
       \ii &= \big( \cT e^{ \int_0^\tau \cL_0 + \cL_A(t)dt } -   e^{ \tau \cL_0} \big) [\rho_S(t_2), O_2] + \left[ \cT \big( e^{ \int_0^\tau \cL_0 + \cL_A(t)dt} -   e^{ \tau \cL_0} \big) \rho_S(t_2), O_2\right] \\ 
        - & \left[ \rho_S(t_2), \big( \cT e^{ \int_0^\tau \cL_0 + \cL_A(t)dt} -   e^{ \tau \cL_0} \big)^\dag O_2 \right] 
      +  i \big( \cT e^{ \int_0^\tau \cL_0 + \cL_B(t)dt} -   e^{ \tau \cL_0} \big) \{\rho_S(t_2), O_2\} \\
      -& i \left\{ \big( \cT e^{ \int_0^\tau \cL_0 + \cL_B(t)dt} -   e^{ \tau \cL_0} \big) \rho_S(t_2), O_2\right\}
        +i \left\{ \rho, \cT \big( e^{ \int_0^\tau \cL_0 + \cL_B(t)dt} -   e^{ \tau \cL_0} \big)^\dag O_2 \right\} \\
        +& \left[ \big( \cT e^{ \int_0^\tau \cL_0 + L_{H_B(t)}dt} -   e^{ \tau \cL_0} \big) \rho_S(t_2),    O_2\right]
        - \left[\rho_S(t_2),  \big(\cT e^{ \int_0^\tau \cL_0 + L_{H_B(t)} dt} -   e^{ \tau \cL_0} \big) O_2\right]
 \end{aligned}
\end{equation}

\end{document}